\documentclass[onecolumn,12pt]{IEEEtran} 
\usepackage{color} 
\usepackage{amsfonts,latexsym}
\usepackage{amssymb}

\newcommand{\lcm}{{\rm lcm}}

\newcommand{\ord}{{\rm ord}}

\newcommand{\bZ}{{\mathbf{Z}}}

\newcommand{\gf}{{\rm GF}}

\newcommand{\C}{{\cal C}} 
\newcommand{\U}{{\cal U}}
\newcommand{\V}{{\mathcal{V}}}
\newcommand{\D}{{\cal D}}

\newtheorem{theorem}{Theorem}

\newtheorem{proposition}[theorem]{Proposition}

\newtheorem{example}{Example}

\def\BibTeX{{\rm B\kern-.05em{\sc i\kern-.025em b}\kern-.08em
    T\kern-.1667em\lower.7ex\hbox{E}\kern-.125emX}}

\begin{document}

\title{Cyclotomic Constructions of Cyclic Codes with Length Being the Product of Two Primes}

\author{Cunsheng Ding\thanks{
C. Ding is with the Department of Computer Science and Engineering, 
The Hong Kong University of Science and Technology, Clear Water Bay, 
Kowloon, Hong Kong. Email: cding@ust.hk}}

\date{\today}
\maketitle

\begin{abstract} 
Cyclic codes are an interesting type of linear codes and have applications in communication 
and storage systems due to their efficient encoding and decoding algorithms. They have been 
studied for decades and a lot of  progress has been made. In this paper, three types of 
generalized cyclotomy of order two and three classes of cyclic codes of length $n_1n_2$ and 
dimension $(n_1n_2+1)/2$ are presented and analysed, where $n_1$ and $n_2$ are two distinct 
primes. Bounds on their minimum odd-like weight are also proved. The three constructions 
produce the best cyclic codes in certain cases. 
\end{abstract}


\section{Introduction} 

Let $q$ be a power of a prime. 
A linear $[n,k,\omega; q]$ code is a $k$-dimensional subspace of $\gf(q)^n$ 
with minimum (Hamming) distance $\omega$. 
Let $A_i$ denote the number of codewords with Hamming weight $i$ in a code 
$\C$ of length $n$. The {\em weight enumerator} of $\C$ is defined by 
$$ 
1+A_1x+A_2x^2+ \cdots + A_nx^n. 
$$

A linear $[n,k]$ code $\C$ over the finite field $\gf(q)$ is called {\em cyclic} if 
$(c_0,c_1, \cdots, c_{n-1}) \in \C$ implies $(c_{n-1}, c_0, c_1, \cdots, c_{n-2}) 
\in \C$.  
Let $\gcd(n, q)=1$. By identifying any vector $(c_0,c_1, \cdots, c_{n-1}) \in \gf(q)^n$ 
with  
$$ 
c_0+c_1x+c_2x^2+ \cdots + c_{n-1}x^{n-1} \in \gf(q)[x]/(x^n-1), 
$$
any code $\C$ of length $n$ over $\gf(q)$ corresponds to a subset of $\gf(q)[x]/(x^n-1)$. 
The linear code $\C$ is cyclic if and only if the corresponding subset in $\gf(q)[x]/(x^n-1)$ 
is an ideal of the ring $\gf(q)[x]/(x^n-1)$. 

Note that every ideal of $\gf(q)[x]/(x^n-1)$ is principal. Let $\C=(g(x))$ be a 
cyclic code. Then $g(x)$ is called the {\em generator polynomial} and 
$h(x)=(x^n-1)/g(x)$ is referred to as the {\em parity-check} polynomial of 
$\C$. 

A vector $(c_0, c_1, \cdots, c_{n-1}) \in \gf(q)^n$ is said to be {\em even-like} 
if $\sum_{i=0}^{n-1} c_i =0$, and is {\em odd-like} otherwise. The minimum weight 
of the even-like codewords, respectively the odd-like codewords of a code is the 
minimum even-like weight, respectively the minimum odd-like weight of the code.  

The error correcting capability of cyclic codes may not be as good as some other linear 
codes. However, cyclic codes have wide applications in storage and communication 
systems because they have efficient encoding and decoding algorithms 
\cite{Chie,Forn,Pran,Rong}. 
For example, Reed–-Solomon codes have found important applications from deep-space 
communication to consumer electronics. They are prominently used in consumer 
electronics such as CDs, DVDs, Blu-ray Discs, in data transmission technologies 
such as DSL \& WiMAX, in broadcast systems such as DVB and ATSC, and in computer 
applications such as RAID 6 systems.

Cyclic codes have been studied for decades and a lot of  progress has been made 
(see for example, \cite{BS06,EL,HPbook,JLX,MoisT,MV99,Ples87,LintW}). 
However, the total number of cyclic codes of length $n$ and dimension $k$ over a 
finite field $\gf(q)$ is in general unknown, not to mention the construction of all of 
them.  An important problem in studying cyclic codes is to find a simple construction 
of the best cyclic codes.  

In this paper, three types of generalized cyclotomy of order two are 
described and three simple constructions of cyclic codes of length 
$n_1n_2$ and dimension $(n_1n_2+1)/2$ are presented, where $n_1$ and 
$n_2$ are two distinct primes. Bounds on their minimum odd-like weight 
are also proved. Some of the codes in this paper are among the best 
cyclic codes. This is the main motivation of investigating the three 
cyclotomic constructions of the cyclic codes in this paper.

\section{All cyclic codes with parameters $[n_1n_2, (n_1n_2+1)/2;q]$} 

Throughout this paper let $n_1$ and $n_2$ be two distinct odd primes, $n=n_1n_2$, and 
let $q$ be a power of a prime such that $\gcd(q, n)=1$. We also define 
$$ 
N=\ord_{n}(q)=\lcm(\ord_{n_1}(q), \ord_{n_2}(q)). 
$$
In this paper we define $\theta=\alpha^{(q^N-1)/n}$, where $\alpha$ is a generator of $\gf(q^N)^*$. 
Hence $\theta$ is an $n$th primitive root of unity in $\gf(q^N)^*$.

\subsection{The general case}
Let $S_0^{(n)}$ be the subgroup 
of $\bZ_n^*$ generated by $q$. Then the cardinality of $S_0^{(n)}$ is $N$. Let $S_i^{(n)}$ 
be all the cosets of the subgroup $S_0^{(n)}$, where $0 \le i < (n_1-1)(n_2-1)/N$. Define 
$$ 
F_j^{(n)}(x)    = \prod_{i \in S_j^{(n)}} (x-\theta^i).   
$$  
It is well known that all $F_j^{(n)}(x)$ are irreducible polynomials over $\gf(q)$ and have  
degree $N$.  

Similarly, for each $j \in \{1,2\}$ let $S_0^{(n_j)}$ be the subgroup of $\bZ_{n_j}^*$ generated 
by $q$. Then the cardinality of $S_0^{(n_j)}$ is $\ord_{n_j}(q)$. Let $S_i^{(n_j)}$ be all the cosets 
of the subgroup $S_0^{(n_j)}$, where $0 \le i < (n_j-1)/\ord_{n_j}(q)$. Let $\theta_j=
\alpha^{(q^N-1)/n_j}$. Define 
$$ 
F_i^{(n_j)}(x)    = \prod_{h \in S_i^{(n_j)}} (x-\theta_{j}^h).   
$$  
It is also well known that all $F_i^{(n_j)}(x)$ are irreducible polynomials over $\gf(q)$ and have  
degree $\ord_{n_j}(q)$.  In addition, we have 
$$ 
x^{n_j}-1=(x-1) \prod_{i=0}^{\frac{n_j-1}{\ord_{n_j}(q)}-1} F_i^{(n_j)}(x) 
$$ 
for each $j \in\{1,2\}$. 

Summarizing the discussions above, we have 
$$ 
x^{n}-1=(x-1) \prod_{i=0}^{\frac{n_1-1}{\ord_{n_1}(q)}-1} F_i^{(n_1)}(x) 
                       \prod_{i=0}^{\frac{n_2-1}{\ord_{n_2}(q)}-1} F_i^{(n_2)}(x) 
                       \prod_{i=0}^{\frac{(n_1-1)(n_2-1)}{\ord_{n}(q)}-1} F_i^{(n)}(x).  
$$

It is hard to give a specific formula for the total number of $[n, (n+1)/2; q]$ cyclic codes. 
However, this number is at least 
$$ 
{\frac{n_1-1}{\ord_{n_1}(q)} \choose \frac{n_1-1}{2\ord_{n_1}(q)} } 
{\frac{n_2-1}{\ord_{n_2}(q)} \choose \frac{n_2-1}{2\ord_{n_2}(q)} }
{\frac{(n_1-1)(n_2-1)}{\ord_{n}(q)} \choose \frac{(n_1-1)(n_2-1)}{2\ord_{n}(q)}}, 
$$  
provided that $(n_j-1)/\ord_{n_j}(q)$ is even for all $j$. 
In many cases, this is indeed the exact number of $[n, (n+1)/2; q]$ cyclic codes. 

In order to show that some of the codes constructed in this paper are the best cyclic codes, 
we provide information about all binary cyclic codes of 
length 119 and dimension 60, and all quaternary codes of length 35 and dimension 18 in the 
next three subsections. 

\subsection{All binary  cyclic codes with length 119 and dimension 60} 

We now consider all binary cyclic codes of length 119 and dimension 60.  
Note that the factorization of $x^{119}-1$ over $\gf(2)$ is 
$$ 
x^{119}-1=f_1(x)f_{31}(x)f_{32}(x)f_{81}(x)f_{82}(x)f_{241}(x)f_{242}(x)f_{243}(x)f_{244}(x), 
$$
where the irreducible polynomials 
\begin{eqnarray*} 
f_1(x) &=& x+1, \\
f_{31}(x) &=&  x^3 + x + 1, \\ 
f_{32}(x) &=& x^3 + x^2 + 1, \\
f_{81}(x) &=&  x^8 + x^5 + x^4 + x^3 + 1, \\ 
f_{82}(x) &=&  x^8 + x^7 + x^6 + x^4 + x^2 + x + 1, \\
f_{241}(x) &=& x^{24} + x^{20} + x^{18} + x^{17} + x^{12} + x^{11} + x^9 + x^7 + x^5 + x^3 + 1, \\
f_{242}(x) &=& x^{24} + x^{21} + x^{19} + x^{17} + x^{15} + x^{13} + x^{12} + x^7 + x^6 + x^4 + 1, \\
f_{243}(x) &=& x^{24} + x^{22} + x^{20} + x^{14} + x^{12} + x^{11} + x^9 + x^8 + x^7 + x^5 + x^2 + x +1, \\
f_{244}(x) &=&  x^{24} + x^{23} + x^{22} + x^{19} + x^{17} + x^{16} + x^{15} + x^{13} + x^{12} + x^{10} + x^4 +x^2 + 1.
\end{eqnarray*} 
Hence, there are altogether 24 binary cyclic codes of length $119$ and dimension $60$. Their generator polynomials 
and minimum nonzero weights are described in Table \ref{tab-mini}. 

\vspace{.25cm}
\begin{table}[ht]
\caption{All binary cyclic codes of length 119 and dimension 60}\label{tab-mini}
\begin{center}
{\begin{tabular}{|l|r|} \hline
Generator Polynomial & Minimum Weight \\\hline \hline
$f_{31}f_{81}f_{241}f_{242}$  &               6  \\ \hline 
$f_{31}f_{81}f_{241}f_{243}$  &              11 \\ \hline 
$f_{31}f_{81}f_{241}f_{244}$  &               8  \\ \hline 
$f_{31}f_{81}f_{242}f_{243}$  &               4  \\ \hline 
$f_{31}f_{81}f_{242}f_{244}$  &              12 \\ \hline 
$f_{31}f_{81}f_{243}f_{244}$  &              12 \\ \hline 
$f_{31}f_{82}f_{241}f_{242}$  &              12   \\ \hline 
$f_{31}f_{82}f_{241}f_{243}$  &               12\\ \hline 
$f_{31}f_{82}f_{241}f_{244}$  &                 8 \\ \hline 
$f_{31}f_{82}f_{242}f_{243}$  &                4 \\ \hline 
$f_{31}f_{82}f_{242}f_{244}$  &               11 \\ \hline 
$f_{31}f_{82}f_{243}f_{244}$  &            6   \\ \hline 
$f_{32}f_{81}f_{241}f_{242}$  &             6    \\ \hline 
$f_{32}f_{81}f_{241}f_{243}$  &             12  \\ \hline 
$f_{32}f_{81}f_{241}f_{244}$  &              4   \\ \hline 
$f_{32}f_{81}f_{242}f_{243}$  &              8   \\ \hline 
$f_{32}f_{81}f_{242}f_{244}$  &            11   \\ \hline 
$f_{32}f_{81}f_{243}f_{244}$  &             12  \\ \hline 
$f_{32}f_{82}f_{241}f_{242}$  &            12     \\ \hline 
$f_{32}f_{82}f_{241}f_{243}$  &             11  \\ \hline 
$f_{32}f_{82}f_{241}f_{244}$  &              4   \\ \hline 
$f_{32}f_{82}f_{242}f_{243}$  &             8    \\ \hline 
$f_{32}f_{82}f_{242}f_{244}$  &            12   \\ \hline 
$f_{32}f_{82}f_{243}f_{244}$  &             6  \\ \hline 
\end{tabular}
}
\end{center}
\end{table}

\subsection{All ternary cyclic codes with length 143 and dimension 72} \label{sec-allternary}

Let $q=3$, $n_1=11$ and $n_2=13$. We have 
$$ 
\ord_{n_1}(q)=5, \ \ord_{n_2}(q)=3, \ \ord_{n}(q)=15. 
$$
The polynomial $x^{143}-1$ is factorized into the product of the following irreducible polynomials: 
\begin{eqnarray*} 
& & x + 2, \\ 
& & x^3 + 2x + 2, \\
& & x^3 + x^2 + 2, \\ 
& & x^3 + x^2 + x + 2, \\
& & x^3 + 2x^2 + 2x + 2, \\
& & x^5 + 2x^3 + x^2 + 2x + 2, \\ 
& & x^5 + x^4 + 2x^3 + x^2 + 2, \\ 
& & x^{15} + x^{12} + x^9 + 2x^8 + x^7 + x^6 + x^5 + 2x^4 + 2x^3 + x^2 + x + 2, \\ 
& & x^{15} + x^{12} + 2x^{11} + x^{10} + 2x^8 + 2x^6 + 2x^5 + x^4 + x^3 + x^2 + 2, \\ 
& & x^{15} + x^{13} + 2x^{11} + x^8 + 2x^6 + x^5 + 2x^4 + x^2 + x + 2, \\ 
& & x^{15} + x^{13} + 2x^{12} + 2x^{11} + 2x^{10} + x^8 + x^7 + x^6 + x^5 + x^2 + 2x + 2, \\ 
& & x^{15} + 2x^{13} + 2x^{12} + 2x^{11} + x^{10} + x^9 + x^7 + 2x^5 + x^4 + 2x^3 + 2, \\ 
& & x^{15} + x^{14} + 2x^{13} + 2x^{10} + 2x^9 + 2x^8 + 2x^7 + x^5 + x^4 + x^3 + 2x^2 + 2, \\ 
& & x^{15} + 2x^{14} + 2x^{13} + x^{11} + 2x^{10} + x^9 + 2x^7 + x^4 + 2x^2 + 2, \\ 
& & x^{15} + 2x^{14} + 2x^{13} + x^{12} + x^{11} + 2x^{10} + 2x^9 + 2x^8 + x^7 + 2x^6 + 2x^3 + 2. 
\end{eqnarray*}

The total number of $[143, 72; 3]$ cyclic codes is thus 
$$ 
{2 \choose 1} {4 \choose 2} {8 \choose 4} = 840. 
$$
Due to the limitation of our computational power, we have not been able to compute the minimum 
weights of the 840 ternary cyclic codes of length 143 and dimension 72.  

\subsection{All quaternary  cyclic codes with length 35 and dimension 18} 

We now consider all quaternary cyclic codes of length 35 and dimension 18.  
Let $w$ be a generator of $\gf(2^2)$ such that $w^2+w+1=0$.  Then the 
factorization of $x^{35}-1$ over $\gf(4)$ is 
$$ 
x^{35}-1=f_1(x)f_{21}(x)f_{22}(x)f_{31}(x)f_{32}(x)f_{61}(x)f_{62}(x)f_{63}(x)f_{64}(x), 
$$
where the irreducible polynomials 
\begin{eqnarray*} 
f_1(x) &=& x+1, \\
f_{21}(x) &=&  x^2 + wx + 1, \\ 
f_{22}(x) &=&  x^2 + w^2x + 1, \\
f_{31}(x) &=&  x^3 + x + 1, \\ 
f_{32}(x) &=&  x^3 + x^2 + 1, \\
f_{61}(x) &=&  x^6 + wx^4 + wx^3 + x^2 + w^2x + 1,  \\
f_{62}(x) &=&  x^6 + w^2x^4 + w^2x^3 + x^2 + wx + 1, \\
f_{63}(x) &=&  x^6 + wx^5 + x^4 + w^2x^3 + w^2x^2 + 1, \\
f_{64}(x) &=&  x^6 + w^2x^5 + x^4 + wx^3 + wx^2 + 1.
\end{eqnarray*} 
Hence, there are altogether 24 quaternary cyclic codes of length $35$ and dimension $18$. Their generator polynomials 
and minimum nonzero weights are described in Table \ref{tab-mini4}. 

\vspace{.25cm}
\begin{table}[ht]
\caption{All quaternary cyclic codes of length 35 and dimension 18}\label{tab-mini4}
\begin{center}
{\begin{tabular}{|l|r|} \hline
Generator Polynomial & Minimum Weight \\\hline \hline
$f_{21}f_{31}f_{61}f_{62}$  &                 4 \\ \hline 
$f_{21}f_{31}f_{61}f_{63}$  &               7 \\ \hline 
$f_{21}f_{31}f_{61}f_{64}$  &                8 \\ \hline 
$f_{21}f_{31}f_{62}f_{63}$  &                4 \\ \hline 
$f_{21}f_{31}f_{62}f_{64}$  &               8 \\ \hline 
$f_{21}f_{31}f_{63}f_{64}$  &               7 \\ \hline 
$f_{21}f_{32}f_{61}f_{62}$  &               7  \\ \hline 
$f_{21}f_{32}f_{61}f_{63}$  &               8 \\ \hline 
$f_{21}f_{32}f_{61}f_{64}$  &               8   \\ \hline 
$f_{21}f_{32}f_{62}f_{63}$  &               4  \\ \hline 
$f_{21}f_{32}f_{62}f_{64}$  &               7 \\ \hline 
$f_{21}f_{32}f_{63}f_{64}$  &              4 \\ \hline 
$f_{22}f_{31}f_{61}f_{62}$  &               4  \\ \hline 
$f_{22}f_{31}f_{61}f_{63}$  &              8  \\ \hline 
$f_{22}f_{31}f_{61}f_{64}$  &              4   \\ \hline 
$f_{22}f_{31}f_{62}f_{63}$  &              8   \\ \hline 
$f_{22}f_{31}f_{62}f_{64}$  &             7  \\ \hline 
$f_{22}f_{31}f_{63}f_{64}$  &             7  \\ \hline 
$f_{22}f_{32}f_{61}f_{62}$  &             7   \\ \hline 
$f_{22}f_{32}f_{61}f_{63}$  &             7  \\ \hline 
$f_{22}f_{32}f_{61}f_{64}$  &             4   \\ \hline 
$f_{22}f_{32}f_{62}f_{63}$  &             8   \\ \hline 
$f_{22}f_{32}f_{62}f_{64}$  &             8  \\ \hline 
$f_{22}f_{32}f_{63}f_{64}$  &             4 \\ \hline 
\end{tabular}
}
\end{center}
\end{table}

\section{A cyclotomy of order two and its codes}\label{sec-W} 

\subsection{A generalization of Whiteman's cyclotomy of order two}

Let $d=\gcd(n_1-1, n_2-1)$, and 
let $g_1$ and $g_2$ be a primitive root of $n_1$ and $n_2$ respectively. Define $g$ by 
\begin{eqnarray}\label{eqn-commonprimi} 
g \equiv g_1 \pmod{n_1}, \ \ 
g \equiv g_2 \pmod{n_2}. 
\end{eqnarray}
By the Chinese Remainder Theorem, $g$ is unique modulo $n$ and 
a common 
primitive root of $n_1$ and $n_2$, and 
$ 
\ord_n(g)=(n_1-1)(n_2-1)/d=e.    
$ 

The following proposition is proved in \cite{White}. 

\begin{proposition} \label{prop-Whiteman}  
Define $\nu$ by  
\begin{eqnarray}\label{eqn-nu}
\nu \equiv g \pmod{n_1}, \ \ 
\nu \equiv 1 \pmod{n_2}. 
\end{eqnarray}
Then  
$$ 
\bZ_n^*=\{g^s\nu^i: s =0,1, \cdots, e-1; i=0, 1, \cdots, d-1\}. 
$$
\end{proposition}  

Whiteman's generalized cyclotomic classes $W_i$ of order $d$ are defined by 
$$ 
W_i=\{g^s\nu^i: s =0,1, \cdots, e-1\}, \ i=0, 1, \cdots, d-1. 
$$
Clearly, $d$ is even. Define two subsets $U_0$ and $U_1$ 
of $\bZ_n^*$ with the cyclotomic classes $W_i$ of order $d$ as 
$$ 
U_0=\bigcup_{i=0}^{(d-2)/2} W_{2i} \mbox{ and }  U_1=\bigcup_{i=0}^{(d-2)/2} W_{2i+1}.  
$$  
Note that $U_0$ is a subgroup of $\bZ_n^*$ and $U_1=\nu U_0$, which is a coset of $U_0$. 
The sets $U_0$ and $U_1$ form a new cyclotomy of order 2, which is different from 
Whiteman's when $d > 2$ and coincides with Whiteman's cyclotomy of order 2 when $d=2$.  

Whiteman proved that $U_0 \cup \{0, n_2, 2n_2, \cdots, (n_1-1)n_2\}$ is a difference set 
when $n_2=n_1+2$, i.e., when they are twin primes \cite{White}. This is the well-known twin-prime 
difference set which has applications in combinatorics, coding theory and communication 
systems. 

The following proposition is proved in \cite{White}. 

\begin{proposition} \label{prop-minusone} 
Let symbols be the same as before. Then  
\begin{eqnarray*}
-1=\left\lbrace 
\begin{array}{l}
g^{e/2} \mbox{ when $(n_1-1)(n_2-1)/d^2$ is odd} \\
g^{t}\nu^{d/2}  \mbox{ when $(n_1-1)(n_2-1)/d^2$ is even,}  
\end{array}
\right.
\end{eqnarray*}
where $t$ is some integer with $0 \le t \le e-1$. 
\end{proposition}

\subsection{The eight cyclic codes from the cyclotomy $(U_0, U_1)$}

Let $D_0^{(n_i)}$ and $D_1^{(n_i)}$ be the set of quadratic residues and nonresidues modulo $n_i$ 
respectively. Define for each $i \in \{0, 1\}$ 
$$ 
d_i^{(n_1)}(x)=\prod_{j \in D_i^{(n_1)}} (x-\theta^{n_2j}), \ 
d_i^{(n_2)}(x)=\prod_{j \in D_i^{(n_2)}} (x-\theta^{n_1j}). 
$$
By definition,  
$$ 
x^{n_j}-1=(x-1)d_0^{(n_j)}(x)  d_1^{(n_j)}(x). 
$$

We define 
\begin{eqnarray}
u_j(x)=\prod_{i \in U_j} (x - \theta^i), \ \ j=0, 1.  
\end{eqnarray} 
and 
\begin{eqnarray}\label{eqn-Wd0d1}
u(x)=u_0(x) u_1(x). 
\end{eqnarray}

Obviously, $u(x) \in \gf(q)[x]$. In fact, we have 
\begin{eqnarray}\label{eqn-Wyuyu1}
x^{n}-1=\prod_{i=0}^{n-1} (x - \theta^i) =\frac{(x^{n_1}-1)(x^{n_2}-1)u(x)}{x-1}. 
\end{eqnarray}

\begin{proposition} \label{prop-july4}
If $q \in U_0$ and $q \bmod n_i \in D_0^{(n_i)}$ for each $i \in \{1, 2\}$, we have $u_i(x) \in \gf(q)[x]$ 
and $d_i^{(n_j)}(x) \in \gf(q)[x]$, and 
$$ 
x^n-1=(x-1) d_0^{(n_1)}(x)  d_0^{(n_2)}(x) u_0(x)  d_1^{(n_1)}(x)  d_1^{(n_2)}(x) u_1(x).  
$$ 
\end{proposition} 

\begin{proof} 
Assume that $q \in U_0$. 
By Proposition \ref{prop-basisgroup}, $qU_i=U_i$ for each $i$. It then follows that 
\begin{eqnarray*}
u_j(x)^q
= \prod_{i \in U_j} (x^q - \theta^{qi})  
= \prod_{i \in qU_j} (x^q - \theta^{i}) 
= \prod_{i \in U_j} (x^q - \theta^{i}) 
= u_j(x^q).     
\end{eqnarray*}  
Similarly, one can prove that $d_i^{(n_j)}(x) \in \gf(q)[x]$ for each $j$. The desired equality then 
follows from the definitions of these polynomials. 
\end{proof} 

Under the conditions that $q \in U_0$ and $q \bmod n_i \in D_0^{(n_i)}$ for each $i \in \{1, 2\}$, 
let $\U_{(i,j,h)}^{(n_1, n_2, q)}$ denote the cyclic code over $\gf(q)$ with generator polynomial 
$u_i(x)d_j^{(n_1)}(x)d_h^{(n_2)}(x)$, where $(i, j, h) \in \{0,1\}^3$. The eight codes clearly have 
length $n$ and dimension $(n+1)/2$.  

\begin{proposition} \label{prop-inconsistent}
There is no $\ell \in U_1$ such that $\ell \bmod{n_i} \in D_1^{(n_i)}$ for all $i \in \{1,2\}$. 
\end{proposition} 

\begin{proof} 
Let $\ell = g^s \nu^i \in U_1$ for some $s$ and $i$ with $0 \le s \le e-1$ and $0 \le i \le d-1$. 
Then $i$ is odd and 
$$ 
\ell \equiv g^{s+i} \pmod{n_1} \mbox{ and } \ell \equiv g^s \pmod{n_2}.  
$$
Since $i$ is odd, it is impossible to have $\ell \bmod{n_i} \in D_1^{(n_i)}$ for all $i \in \{1,2\}$ 
at the same time. 
\end{proof} 

Proposition \ref{prop-inconsistent} means that it is impossible to prove a square-root bound on 
the minimum odd-like weight of the codes with the traditional argument for the square-root bound  
of quadratic residue codes. Hence, we develop another type of  bound. 

\begin{proposition}\label{prop-sqrtnewb} 
Let $\omega_{(i,j,h)}^{(n)}$ denote the minimum odd-like weight of the code $\U_{(i,j,h)}^{(n_1, n_2, q)}$ 
and let  $\omega_{i}^{(n/n_j)}$ denote the minimum odd-like weight of the cyclic code of length $n$  
over $\gf(q)$ generated by the polynomial $(x^n-1)/(x-1)d_i^{(n_j)}(x)$ for all $i \in \{0,1\}$ and all 
$j \in \{0,1\}$.   
We have then 
$$ 
\omega_{(0,0,0)}^{(n)} \ge \sqrt{\max(\omega_{1}^{(n/n_1)}, \omega_{1}^{(n/n_2)})}. 
$$
\end{proposition} 

\begin{proof} 
Let $\ell = g^s x^i$ for some $s$ and $i$. Note that 
$$ 
\ell \equiv g^{s+i} \pmod{n_1} \mbox{ and } \ell \equiv g^s \pmod{n_2}. 
$$
We have then the following conclusions.  
\begin{enumerate} 
\item There is an $\ell_1 \in U_1$ such that $\ell_1 \bmod n_1 \in D_{0}^{(n_1)}$ and $\ell_1 \bmod n_2 \in D_{1}^{(n_1)}$. 
\item There is an $\ell_2 \in U_1$ such that $\ell_2 \bmod n_1 \in D_{1}^{(n_1)}$ and $\ell_2 \bmod n_2 \in D_{0}^{(n_1)}$. 
\end{enumerate} 
Let $a(x)$ be a codeword in $\U_{(0,0,0)}^{(n_1, n_2, q)}$ with minimum odd-like weight  $\omega_{(0,0,0)}^{(n)} $. 
Then $a(x^{\ell_1})$ is a codeword in $\U_{(1,0,1)}^{(n_1, n_2, q)}$ with minimum odd-like weight  $\omega_{(0,0,0)}^{(n)}$, 
and $a(x^{\ell_2})$ is a codeword in $\U_{(1,1,0)}^{(n_1, n_2, q)}$ with minimum odd-like weight  $\omega_{(0,0,0)}^{(n)}$. 
It follows that $a(x)a(x^{\ell_1})$ is an odd-like codeword in the cyclic code with the generator polynomial $(x^n-1)/(x-1)d_1^{(n_1)}(x)$ 
and $a(x)a(x^{\ell_2})$ is an odd-like codeword in the cyclic code with the generator polynomial $(x^n-1)/(x-1)d_1^{(n_2)}(x)$. Since 
$a(x)a(x^{\ell_1})$ and  $a(x)a(x^{\ell_2})$ have at most $(\omega_{(0,0,0)}^{(n)})^2$ terms, the desired lower bound then 
follows. 
\end{proof} 

Similar bounds for other $\omega_{(i,j,h)}^{(n)}$ can be written down.  The lower bound of Proposition \ref{prop-sqrtnewb} 
depends on the minimum odd-like weight of two special codes and may not be convenient to use. It will be seen later that the BCH 
bound on these codes could be much better. 

We have the following remarks on Whiteman's cyclotomy of order two and its codes defined above. 
\begin{enumerate} 
\item In \cite{DH99} Ding and Helleseth generalized Whiteman's cyclotomy of order two into the 
         case that $n=\prod_{i=1}^{t} n_i^{e_i}$, where all $n_i$ are pairwise distinct primes and  
         $\gcd(n_i-1, n_j-1)=2$ for all pairs of distinct $i$ and $j$, and introduced eight binary 
         cyclotomic codes. 
\item However, in the case that $n=n_1n_2$ the cyclotomy of order two in \cite{DH99} is just a 
          special case of Whiteman's cyclotomy of order two because of the required condition 
          $\gcd(n_1-1, n_2-1)=2$. Hence the cyclotomic codes defined in \cite{DH99} are special 
          cases of the eight codes over $\gf(q)$ in this section when $q=2$ and $\gcd(n_1-1, n_2-1)=2$.    
          So Whiteman's cyclotomy of order two yields more codes. For example, when 
          $(n_1, n_2)=(17, 41)$ Whiteman's cyclotomy of order two gives eight binary cyclic codes, 
          while the cyclotomy introduced in \cite{DH99} does not work for this pair of $n_1$ and $n_2$.           
\end{enumerate}

\subsection{The binary case} 

\begin{proposition}\label{prop-twodd}
The integer $2 \in U_0$ if and only if 
$$ 
n_1 \equiv \pm 1 \pmod{8} \mbox{ and } n_2 \equiv \pm 1 \pmod{8} 
$$
or 
$$ 
n_1 \equiv \pm 3 \pmod{8} \mbox{ and } n_2 \equiv \pm 3 \pmod{8} .
$$
\end{proposition} 

\begin{proof} 
Recall that $\bZ_n^*=U_0 \cup U_1$. By Proposition \ref{prop-Whiteman}, there are two integers 
$0 \le s \le e-1$ and  $0 \le i \le d-1$ such that $2=g^s\nu^i$.  
It then follows from (\ref{eqn-commonprimi}) that 
$$ 
g^{s+i} \equiv 2 \pmod{n_1} \mbox{ and } g^{s} \equiv 2 \pmod{n_2}.   
$$ 
Hence, $i$ is even if and only if the one of the conditions in this proposition is satisfied. 
Note that $2 \in U_0$ if and only if $i$ is even. The proof is then completed. 
\end{proof} 

Therefore, in the case that $n_1 \equiv \pm 1 \pmod{8}$ and $n_2 \equiv \pm 1 \pmod{8}$, 
we have indeed the eight binary cyclic codes $\U_{(i,j,h)}^{(n_1, n_2, 2)}$. 

\begin{example} 
When $(n_1, n_2, q)=(7, 17, 2)$,  we have 
\begin{eqnarray*}
U_0 = \left\{ \begin{array}{l} 
1, 2, 3, 4, 5, 6, 8, 9, 10, 12, 15, 16, 18, 20, 24, 25, 27, 30, 31, 32,
36, 40, 41, 43, 45, 48, \\ 50, 53, 54, 60, 61, 62, 64, 67, 72, 73, 75, 80, 81, 82,
86, 90, 93, 96, 97, 100, 106, 108 
                 \end{array}                  
                 \right\}, \\
U_1 = \left\{ \begin{array}{l} 
11, 13, 19, 22, 23, 26, 29, 33, 37, 38, 39, 44, 46, 47, 52, 55, 57, 58,
59, 65, 66, 69, 71, 74, 76, 78, \\ 79, 83, 87, 88, 89, 92, 94, 95, 99, 101, 103,
104, 107, 109, 110, 111, 113, 114, 115, 116, 117, 118 
                 \end{array}                  
                 \right\}                   
\end{eqnarray*} 
and 
\begin{eqnarray*}
d_0^{(n_1)}(x) &=& x^3+x+1, \\
d_1^{(n_1)}(x) &=& x^3+x^2+1, \\ 
d_0^{(n_2)}(x) &=& x^8 + x^7 + x^6 + x^4 + x^2 + x + 1, \\
d_1^{(n_2)}(x) &=& x^8 + x^5 + x^4 + x^3 + 1, \\ 
u_0(x) &=& x^{48} + x^{46} + x^{41} + x^{39} + x^{37} + x^{33} + x^{30} + x^{29} + x^{27} + x^{25} + x^{23} +
    x^{18} + x^{17} + x^{13} + \\ & & x^{10} + x^9 + x^8 + x^7 + x^6 + x^5 + x^4 + x^3 + x^2 +
    x + 1,  \\
u_1(x) &=& x^{48} + x^{47} + x^{46} + x^{45} + x^{44} + x^{43} + x^{42} + x^{41} + x^{40} + x^{39} + x^{38} +
    x^{35} + x^{31} + x^{30} + \\ & & x^{25} + x^{23} + x^{21} + x^{19} + x^{18} + x^{15} + x^{11} + x^9 +
    x^7 + x^2 + 1. 
\end{eqnarray*} 

The minimum nonzero weights of the eight codes are given in Table \ref{tab-W1}. 
Four of the eight codes are the best binary cyclic codes of length 119 and dimension 60 and have 
minimum weight 12 according to Table \ref{tab-mini}, and the remaining four have minimum 
weight $11$.  
\end{example} 

\vspace{.25cm}
\begin{table}[ht]
\caption{The binary cyclic codes of length 119 and dimension 60 from Whiteman's cyclotomy}\label{tab-W1}
\begin{center}
{\begin{tabular}{|l|r|} \hline
The code & Minimum Weight \\\hline \hline
$\U_{(0,0,0)}^{(7, 17, 2)}$   &             12    \\ \hline 
$\U_{(1,0,0)}^{(7, 17, 2)}$  &             11    \\ \hline 
$\U_{(0,1,0)}^{(7, 17, 2)}$  &             11    \\ \hline 
$\U_{(0,0,1)}^{(7, 17, 2)}$  &             11    \\ \hline 
$\U_{(1,1,0)}^{(7, 17, 2)}$  &             12    \\ \hline
$\U_{(1,0,1)}^{(7, 17, 2)}$  &             12    \\ \hline 
$\U_{(0,1,1)}^{(7, 17, 2)}$  &             12    \\ \hline
$\U_{(1,1,1)}^{(7, 17, 2)}$  &             11    \\ \hline
\end{tabular}
}
\end{center}
\end{table}

\subsection{The ternary case} 

\begin{proposition}\label{prop-threed}
The integer $3 \in U_0$ and $3 \bmod{n_i} \in D_0^{(n_i)}$ for all $i$ if and only if 
$$ 
n_1 \equiv \pm 1 \pmod{12} \mbox{ and } n_2 \equiv \pm 1 \pmod{12} . 
$$
\end{proposition} 

\begin{proof} 
With the Law of Quadratic Reciprocity, one can prove that  $3 \bmod{n_i} \in D_0^{(n_i)}$ 
if and only if $n_i \equiv \pm 1 \pmod{12}$. In addition, let $3=g^s \nu^i$ for some $s$ and $i$. Then 
we have 
$$ 
3 \equiv g^{s+i} \pmod{n_1} \mbox{ and } 3 \equiv g^{s} \pmod{n_2}.  
$$ 
If $3$ is a quadratic residue modulo both $n_1$ and $n_2$, both $s$ and $s+i$ must be even. 
It follows that $i$ must be even. Hence, $3 \in U_0$. 
\end{proof} 

So in the case that $n_1 \equiv \pm 1 \pmod{12}$ and $n_2 \equiv \pm 1 \pmod{12}$, 
we have indeed the eight ternary cyclic codes $\U_{(i,j,h)}^{(n_1, n_2, 3)}$. 

\begin{example} 
When $(n_1, n_2, q)=(11, 13, 3)$, the minimum nonzero weights of the eight ternary codes 
are given in Table \ref{tab-W2}. 
Four of the eight codes have 
minimum weight 12, and the remaining four have minimum 
weight $11$.  
\end{example} 

\vspace{.25cm}
\begin{table}[ht]
\caption{The ternary cyclic codes of length 143 and dimension 72 from Whiteman's cyclotomy}\label{tab-W2}
\begin{center}
{\begin{tabular}{|l|r|} \hline
The code & Minimum Weight \\\hline \hline
$\U_{(0,0,0)}^{(11, 13, 3)}$   &             12    \\ \hline 
$\U_{(1,0,0)}^{(11, 13, 3)}$  &             11    \\ \hline 
$\U_{(0,1,0)}^{(11, 13, 3)}$  &             11    \\ \hline 
$\U_{(0,0,1)}^{(11, 13, 3)}$  &             11    \\ \hline 
$\U_{(1,1,0)}^{(11, 13, 3)}$  &             12    \\ \hline
$\U_{(1,0,1)}^{(11, 13, 3)}$  &             12    \\ \hline 
$\U_{(0,1,1)}^{(11, 13, 3)}$  &             12    \\ \hline
$\U_{(1,1,1)}^{(11, 13, 3)}$  &             11    \\ \hline
\end{tabular}
}
\end{center}
\end{table}

\subsection{The quaternary case} 

The following proposition can be similarly proved. 

\begin{proposition}\label{prop-fourdd}
The integer $4 \in U_0$ and $4 \bmod{n_j} \in D_0^{(n_j)}$ for all $j$ if and only if 
$$ 
n_1 \equiv \pm 1 \pmod{4} \mbox{ and } n_2 \equiv \pm 1 \pmod{4} . 
$$
\end{proposition}

Therefore, in the case that $n_1 \equiv \pm 1 \pmod{4}$ and $n_2 \equiv \pm 1 \pmod{4}$, 
we have indeed the eight quaternary cyclic codes $\U_{(i,j,h)}^{(n_1, n_2, 4)}$. 

\begin{example} 
When $(n_1, n_2, q)=(5, 7, 4)$,  
the minimum nonzero weights of the eight quaternary codes are given in Table \ref{tab-4W1}. 
Four of the eight codes are the best quaternary cyclic codes of length 35 and dimension 18 and have 
minimum weight 8 according to Table \ref{tab-mini4}, and the remaining four have minimum 
weight $7$. 
\end{example} 

\vspace{.25cm}
\begin{table}[ht]
\caption{The quaternary cyclic codes of length 35 and dimension 18 from Whiteman's cyclotomy}\label{tab-4W1}
\begin{center}
{\begin{tabular}{|l|r|} \hline
The code & Minimum Weight \\\hline \hline
$\U_{(0,0,0)}^{(5, 7, 4)}$   &             8    \\ \hline 
$\U_{(1,0,0)}^{(5, 7, 4)}$  &             7    \\ \hline 
$\U_{(0,1,0)}^{(5, 7, 4)}$  &             7    \\ \hline 
$\U_{(0,0,1)}^{(5, 7, 4)}$  &             7    \\ \hline 
$\U_{(1,1,0)}^{(5, 7, 4)}$  &             8    \\ \hline
$\U_{(1,0,1)}^{(5, 7, 4)}$  &             8    \\ \hline 
$\U_{(0,1,1)}^{(5, 7, 4)}$  &             8    \\ \hline
$\U_{(1,1,1)}^{(5, 7, 4)}$  &             7    \\ \hline
\end{tabular}
}
\end{center}
\end{table}

\subsection{The robustness of Whiteman's cyclotomy of order two} 

All the binary cyclic codes of length 119 and dimension 60 and all the quaternary cyclic codes 
of length 35 and dimension 18 presented before are either the best or almost the best in 
terms of their minimum weights. In this section, we will provide theoretical evidences for 
this fact. We consider only the case that $(n_1, n_2)=(7, 17)$. 

Let $(n_1, n_2, q)=(7, 17, q)$, and let $I(i, j, h)$ denote the set of 
zeros $\theta^l$ of the generator polynomial of the code $\U_{(i,j,h)}^{(7, 17, q)}$. 
Then the set of exponents $l$ of $\theta^l$ in $I(i, j, h)$ is the following. 
\begin{itemize}
\item When $(i, j, h)=(0, 0, 0)$, the set is 
\begin{eqnarray*} 
\left\{ \begin{array}{l}
1, 2, 3, 4, 5, 6, 7, 8, 9, 10, 12, 14, 15, 16, 17, 18, 20, 24, 25, 27,28, 30, \\
31, 32, 34, 36, 40, 41, 43, 45, 48, 50, 53, 54, 56, 60, 61, 62, 63, 64, 67, \\ 
68, 72, 73, 75, 80, 81, 82, 86, 90, 91, 93, 96, 97, 100, 105, 106, 108, 112
\end{array}
\right\}. 
\end{eqnarray*} 
The BCH bound says that the minimum weight of the code is at least 11. 
The case that $(i, j, h)=(1, 1, 1)$ is equivalent to this case as any element 
in $D_1$ times this set gives the set for the case $(i, j, h)=(1, 1, 1)$. So 
we have the same lower bound for the code  $\U_{(1,1,1)}^{(7, 17, q)}$.

\item When $(i, j, h)=(1, 0, 0)$, the set is 
\begin{eqnarray*} 
\left\{ \begin{array}{l}
7, 11, 13, 14, 17, 19, 22, 23, 26, 28, 29, 33, 34, 37, 38, 39, 44, 46, 47, 52, 55, \\ 
56, 57, 58, 59, 63, 65, 66, 68, 69, 71, 74, 76, 78, 79, 83, 87, 88, 89, 91, 92, 94, \\ 
95, 99, 101, 103, 104, 105, 107, 109, 110, 111, 112, 113, 114,
115, 116, 117, 118
\end{array}
\right\}. 
\end{eqnarray*} 
The BCH bound says that the minimum weight of the code is at least 11. 
The case that $(i, j, h)=(0, 1, 1)$ is equivalent to this case. So 
we have the same lower bound for the code  $\U_{(0,1,1)}^{(7, 17, q)}$.

\item When $(i, j, h)=(0, 1, 0)$, the set is 
\begin{eqnarray*} 
\left\{ \begin{array}{l}
1, 2, 3, 4, 5, 6, 7, 8, 9, 10, 12, 14, 15, 16, 18, 20, 24, 25, 27, 28, 30, 31, 32, 36, \\ 
40, 41, 43, 45, 48, 50, 51, 53, 54, 56, 60, 61, 62, 63, 64, 67, 72, 73, 75, 80, 81, \\ 
82, 85, 86, 90, 91, 93, 96, 97, 100, 102, 105, 106, 108, 112
\end{array}
\right\}. 
\end{eqnarray*} 
The BCH bound says that the minimum weight of the code is at least 11. 
The case that $(i, j, h)=(1, 0, 1)$ is equivalent to this case. So 
we have the same lower bound for the code  $\U_{(1,0,1)}^{(7, 17, q)}$.

\item When $(i, j, h)=(0, 0, 1)$, the set is 
\begin{eqnarray*} 
\left\{ \begin{array}{l}
1, 2, 3, 4, 5, 6, 8, 9, 10, 12, 15, 16, 17, 18, 20, 21, 24, 25, 27, 30, 31, 32, 34, 35, \\ 
36, 40, 41, 42, 43, 45, 48, 49, 50, 53, 54, 60, 61, 62, 64, 67, 68, 70, 72, 73, 75, 77, \\ 
80, 81, 82, 84, 86, 90, 93, 96, 97, 98, 100, 106, 108
\end{array}
\right\}. 
\end{eqnarray*} 
The BCH bound says that the minimum weight of the code is at least 11. 
The case that $(i, j, h)=(1, 1, 0)$ is equivalent to this case. So 
we have the same lower bound for the code  $\U_{(1,1,0)}^{(7, 17, q)}$.

\end{itemize}    

The lower bound 11 is true not only for $q=2$, but also for $q=4$. This experimental data 
explains why the codes are the best or almost the best among all cyclic codes of the same 
length and the same  dimension.

\section{An extension of an earlier generalized cyclotomy of order two and its codes}\label{sec-D} 

\subsection{The extension of an earlier generalized cyclotomy of order two}\label{sec-cyclDH} 

Since $n_1$ and $n_2$ are odd primes, $d$ must be 
even. It is easily seen that $e=(n_1-1)(n_2-1)/d$ is also even. 
Define 
\begin{eqnarray*} 
D_0 &=& \{g^{2s} \nu^i: s=0, 1, \cdots, (e-2)/2; i=0,1, \cdots, d-1\}, \\
D_1 &=& \{g^{2s+1} \nu^i: s=0, 1, \cdots, (e-2)/2; i=0,1, \cdots, d-1\}. 
\end{eqnarray*} 
Clearly, $D_1=gD_0$, and $D_0$ and $D_1$ form a partition of $\bZ_n^*$.

A proof of the following proposition is straightforward and is omitted. 

\begin{proposition} \label{prop-basisgroup} 
Let symbols be the same as before. 
\begin{enumerate} 
\item $D_0$ is a subgroup of $\bZ_n^*$ and has order $(n_1-1)(n_2-1)/2$. 
\item If $a \in D_0$, we have $aD_i=D_i$. If $a \in D_1$,  we have $aD_i=D_{(i+1) \bmod 2}$. 
\end{enumerate} 
\end{proposition} 

The sets $D_0$ and $D_1$ are the cyclotomic classes of order 2, and are clearly 
different from Whiteman's cyclotomic classes of order 2 described in Section \ref{sec-W}. 
We point out here that this generalized cyclotomy of order two is the same as the one 
introduced by Ding and Helleseth when $\gcd(n_1-1, n_2-1)=2$, and is indeed an extension.   
We will use $D_0$ and $D_1$ to describe cyclic codes in the sequel.

\subsection{The construction of eight cyclic codes} 

Let $\theta$ and other symbols be the same as before. 
Define 
\begin{eqnarray}
d_i(x)=\prod_{j \in D_i} (x - \theta^j). 
\end{eqnarray}

\begin{proposition} 
If $q \in D_0$ and $q \bmod n_i \in D_0^{(n_i)}$ for each $i \in \{1, 2\}$, we have $d_i(x) \in \gf(q)[x]$ 
and $d_i^{(n_j)}(x) \in \gf(q)[x]$, and 
$$ 
x^n-1=(x-1) d_0^{(n_1)}(x)  d_0^{(n_2)}(x) d_0(x)  d_1^{(n_1)}(x)  d_1^{(n_2)}(x) d_1(x).  
$$ 
\end{proposition} 

\begin{proof} 
The proof of Proposition \ref{prop-july4} is easily modified into a proof for this proposition. 
\end{proof} 

Under the conditions that $q \in D_0$ and $q \bmod n_i \in D_0^{(n_i)}$ for each $i \in \{1, 2\}$, 
let $\D_{(i,j,h)}^{(n_1, n_2, q)}$ denote the cyclic code over $\gf(q)$ with generator polynomial 
$d_i(x)d_j^{(n_1)}(x)d_h^{(n_2)}(x)$, where $(i, j, h) \in \{0,1\}^3$. The eight codes clearly have 
length $n$ and dimension $(n+1)/2$.  

\begin{proposition} \label{prop-consistent}
There is an $\ell \in U_1$ such that $\ell \bmod{n_i} \in D_1^{(n_i)}$ for all $i \in \{1,2\}$. 
\end{proposition} 

\begin{proof} 
Take any $\ell = g^s \nu^i \in U_1$ for any odd $s$ and even $i$ with $0 \le s \le e-1$ and $0 \le i \le d-1$. 
Note that $i$ is even and 
$$ 
\ell \equiv g^{s+i} \pmod{n_1} \mbox{ and } \ell \equiv g^s \pmod{n_2}.  
$$
We have that $\ell \bmod{n_i} \in D_1^{(n_i)}$ for all $i \in \{1,2\}$ 
at the same time. 
\end{proof} 

Due to Proposition \ref{prop-consistent} we can prove a square-root bound on 
the minimum odd-like weight of the codes. 

\begin{theorem}\label{thm-srbDH} 
For each $(i,j,h) \in \{0,1\}^3$ 
the code  $\D_{(i,j,h)}^{(n_1, n_2, q)}$ has parameters $[n, (n+1)/2]$. Let $d_{(i,j,h)}$ denote the minimum 
odd-like weight 
in $\D_{(i,j,h)}^{(n_1, n_2, q)}$. Then 
\begin{itemize} 
\item $d_{(i,j,h)} \ge \sqrt{n}$, and 
\item $d_{(i,j,h)}^2 - d_{(i,j,h)} + 1 \ge n$ if $n_1 \equiv -1 \pmod{8}$ and $n_2 \equiv -1 \pmod{8}$. 
\end{itemize} 
\end{theorem}     

\begin{proof} 
We first prove that  $\D_{(i,j,h)}^{(n_1, n_2, q)}$ and  $\D_{((i+1) \bmod{2}, (j+1) \bmod{2}, (h+1) \bmod{2})}^{(n_1, n_2, q)}$ 
are equivalent. 
Note that $g \in D_1$ and $g$ is a quadratic nonresidue modulo both $n_1$ and $n_2$. 
Let $\ell \in D_1$ such that $\ell$ is a quadratic nonresidue modulo both $n_1$ and $n_2$.
The permutation 
of coordinates in $\gf(2)[x]/(x^n-1)$ induced by $x \mapsto x^\ell$  
interchanges $\D_{(i,j,h)}^{(n_1, n_2, q)}$ and  $\D_{((i+1) \bmod{2}, (j+1) \bmod{2}, (h+1) \bmod{2})}^{(n_1, n_2, q)}$. This proves the equivalence. 
Hence the two codes have the same minimum nonzero weight 
and the same minimum odd-like weight.  

Let $a(x)$ be a codeword of minimum odd-like weight $d_{(i,j,h)}$ in $\D_{(i,j,h)}^{(n_1, n_2, q)}$. 
Hence, $\hat{a}(x)=a(x^\ell)$ is a codeword of minimum odd-like weight  $d_{(i,j,h)}$ 
in $\D_{((i+1) \bmod{2}, (j+1) \bmod{2}, (h+1) \bmod{2})}^{(n_1, n_2, q)}$.   Then $a(x)\hat{a}(x)$ must be in 
$\D_{(i,j,h)}^{(n_1, n_2, q)} \cap \D_{((i+1) \bmod{2}, (j+1) \bmod{2}, (h+1) \bmod{2})}^{(n_1, n_2, q)}$, i.e., 
is a multiple of 
$$ 
\frac{x^n-1}{x-1}=x^{n-1}+x^{n-2}+\cdots +x+1.
$$ 
Thus, $a(x)\hat{a}(x)$ has weight $n$. Since $a(x)$ has weight 
$d_{(i,j,h)}$, the maximum number of coefficients in $a(x)\hat{a}(x)$ is $d_{(i,j,h)}^2$. 
Therefore, $d_{(i,j,h)}^2 \ge n$. 

If $n_1 \equiv -1 \pmod{8}$ and $n_2 \equiv -1 \pmod{8}$, by Proposition \ref{prop-minusone} we can take $\ell =-1$. 
In this case, the maximum number of 
coefficients in $a(x)\hat{a}(x)$ is $d_{(i,j,h)}^2-d_{(i,j,h)}^2+1$. 
Therefore, $d_{(i,j,h)}^2-d_{(i,j,h)}^2+1 \ge n$. 
\end{proof}

Note that the cyclotomy of order two described in this section is an extension of the one given in 
\cite{DH98}. The codes described in \cite{DH98} are only special cases of the codes of this section. 
First of all, the codes in \cite{DH98} require the condition that $\gcd(n_1-1, n_2-1)=2$, while the 
eight codes in this section do not require this condition. Secondly, the codes in \cite{DH98} are 
binary only, while the eight codes in this section are over $\gf(q)$.

\subsection{The binary case}\label{sec-Ding}

The following proposition will be useful later. 

\begin{proposition} 
Assume that $n_1 \equiv \pm 1 \pmod{8}$ and $n_2 \equiv \pm 1 \pmod{8}$. Then 
$-1 \in D_1$ if and only if $n_2 \equiv -1 \pmod{8}$. 
\end{proposition} 

\begin{proof} 
Let $-1=g^s \nu^i$ for some fixed $0 \le s \le e-1$ and $0 \le i \le d-1$.  By (\ref{eqn-nu}), 
$$ 
-1 \equiv g^{s+i} \pmod{n_1}, \ \ -1 \equiv g^{s} \pmod{n_2}. 
$$ 
Hence $-1 \in D_1$ if and only if $s$ is odd, which is equivalent to $n_2 \equiv -1 \pmod{8}$. 
\end{proof}

We will need the following proposition in the sequel. 

\begin{proposition}\label{prop-twod}
The integer $2 \in D_0$ if and only if 
$
n_2 \equiv \pm 1 \pmod{8} . 
$
\end{proposition} 

\begin{proof} 
Recall that $\bZ_n^*=D_0 \cup D_1$. By Proposition \ref{prop-Whiteman}, there are two integers 
$0 \le s \le e-1$ and  $0 \le i \le d-1$ such that $2=g^s\nu^i$.  
It then follows from (\ref{eqn-commonprimi}) that 
$$ 
g^{s+i} \equiv 2 \pmod{n_1} \mbox{ and } g^{s} \equiv 2 \pmod{n_2}.   
$$ 
Hence, $s$ is even if and only if  $n_2 \equiv \pm 1 \pmod{8}$ . 
Note that $2 \in D_0$ if and only if $s$ is even. The proof is then completed. 
\end{proof} 

If $n_1 \equiv \pm 1 \pmod{8}$ and $n_2 \equiv \pm 1 \pmod{8}$, then we have the 
following factorization of $x^n-1$ over $\gf(2)$: 
$$
x^n-1=(x-1) d_0^{(n_1)}(x)  d_0^{(n_2)}(x) d_0(x)  d_1^{(n_1)}(x)  d_1^{(n_2)}(x) d_1(x),    
$$
and thus the eight binary cyclic codes $\D_{(i,j,h)}^{(n_1, n_2, 2)}$.  

\begin{example}\label{exam-2} 
Let $(n_1, n_2, q)=(7, 17, 2)$. In this case we have 
\begin{eqnarray*}
D_0 &=& \left\{ \begin{array}{l} 
                 1, 2, 4, 8, 9, 13, 15, 16, 18, 19, 25, 26, 30, 32, 33, 36, 38, 43,  
                 47, 50, 52, 53, 55, 59, 60, 64, 66, \\ 67, 69, 72, 76, 81, 83, 86,  
                 87, 89, 93, 94, 100, 101, 103, 104, 106, 110, 111, 115, 117, 118 
                 \end{array}                  
                 \right\}, \\
D_1 &=& \left\{ \begin{array}{l} 
                3, 5, 6, 10, 11, 12, 20, 22, 23, 24, 27, 29, 31, 37, 39, 40, 41, 44,  
                45, 46, 48, 54, 57, 58, 61, 62, \\ 65, 71, 73, 74, 75, 78, 79, 80, 82, 88,  
                90, 92, 95, 96, 97, 99, 107, 108, 109, 113, 114, 116 \\  
                 \end{array}                  
                 \right\}                   
\end{eqnarray*} 
and 
\begin{eqnarray*}
d_0^{(n_1)}(x) &=& x^3+x+1, \\
d_1^{(n_1)}(x) &=& x^3+x^2+1, \\ 
d_0^{(n_2)}(x) &=& x^8 + x^7 + x^6 + x^4 + x^2 + x + 1, \\
d_1^{(n_2)}(x) &=& x^8 + x^5 + x^4 + x^3 + 1, \\ 
d_0(x) &=& x^{48} + x^{45} + x^{44} + x^{43} + x^{42} + x^{41} + x^{39} +  x^{38} + x^{37} + x^{36} + x^{35} +  
       x^{34} + x^{32} + x^{30} + \\ & & x^{29} + x^{27} + x^{26} + x^{24} + x^{22} + x^{21} + x^{19} + x^{18} + x^{16} +  
          x^{14} + x^{13} + x^{12} + x^{11} + x^{10} + \\ 
          & & x^9 + x^7 + x^6 + x^5 + x^4 + x^3+ 1, \\
d_1(x) &=& x^{48} + x^{47} + x^{45} + x^{42} + x^{40} + x^{39} + x^{31} + x^{30} + x^{28} + x^{27} + x^{26} + x^{25} +  
          x^{24} + x^{23} + \\ & & x^{22} + x^{21} + x^{20} + x^{18} + x^{17} + x^9 + x^8 + x^6 +    x^3 + x + 1. 
\end{eqnarray*} 
The eight binary cyclic codes and their minimum weights are depicted in Table \ref{tab-D1}. 
In this example, half of the codes are the best cyclic codes of length 119 and dimension 
60 according to Table \ref{tab-mini}. 
\end{example}

\vspace{.25cm}
\begin{table}[ht]
\caption{The binary cyclic codes of length 119 and dimension 60 from the cyclotomy of Section \ref{sec-cyclDH}}\label{tab-D1}
\begin{center}
{\begin{tabular}{|l|r|} \hline
Generator Polynomial & Minimum Weight \\\hline \hline
$\D_{(0,0,0)}^{(7, 17, 2)}$     &             12    \\ \hline 
$\D_{(1,0,0)}^{(7, 17, 2)}$   &               6 \\ \hline 
$\D_{(0,1,0)}^{(7, 17, 2)}$  &              12   \\ \hline 
$\D_{(0,0,1)}^{(7, 17, 2)}$   &               6  \\ \hline 
$\D_{(1,1,0)}^{(7, 17, 2)}$   &               6  \\ \hline
$\D_{(1,0,1)}^{(7, 17, 2)}$   &             12    \\ \hline 
$\D_{(0,1,1)}^{(7, 17, 2)}$   &               6 \\ \hline
$\D_{(1,1,1)}^{(7, 17, 2)}$  &             12    \\ \hline
\end{tabular}
}
\end{center}
\end{table}

\subsection{The ternary case} 

The proof of Proposition \ref{prop-threed} can be slightly modified into a proof of the 
following.  

\begin{proposition}\label{prop-threed2}
The integer $3 \in D_0$ and $3 \bmod{n_i} \in D_0^{(n_i)}$ for all $i$ if and only if 
$$ 
n_1 \equiv \pm 1 \pmod{12} \mbox{ and } n_2 \equiv \pm 1 \pmod{12} . 
$$
\end{proposition}

Hence, in the case that $n_1 \equiv \pm 1 \pmod{12}$ and $n_2 \equiv \pm 1 \pmod{12}$, 
we have indeed the eight ternary cyclic codes $\D_{(i,j,h)}^{(n_1, n_2, 3)}$. 

\begin{example} 
When $(n_1, n_2, q)=(11, 13, 3)$, the minimum nonzero weights of the eight ternary codes 
are given in Table \ref{tab-W23}. 
Four of the eight codes have 
minimum weight 12, and the remaining four have minimum 
weight $6$.  
\end{example} 

\vspace{.25cm}
\begin{table}[ht]
\caption{The ternary cyclic codes of length 143 and dimension 72 from the cyclotomy of Section \ref{sec-cyclDH}}\label{tab-W23}
\begin{center}
{\begin{tabular}{|l|r|} \hline
The code & Minimum Weight \\\hline \hline
$\D_{(0,0,0)}^{(11, 13, 3)}$   &          12       \\ \hline 
$\D_{(1,0,0)}^{(11, 13, 3)}$  &           6      \\ \hline 
$\D_{(0,1,0)}^{(11, 13, 3)}$  &           12      \\ \hline 
$\D_{(0,0,1)}^{(11, 13, 3)}$  &           6      \\ \hline 
$\D_{(1,1,0)}^{(11, 13, 3)}$  &          6       \\ \hline
$\D_{(1,0,1)}^{(11, 13, 3)}$  &           12      \\ \hline 
$\D_{(0,1,1)}^{(11, 13, 3)}$  &          6       \\ \hline
$\D_{(1,1,1)}^{(11, 13, 3)}$  &          12       \\ \hline
\end{tabular}
}
\end{center}
\end{table}

\subsection{The quaternary case}

The following proposition can be similarly proved. 

\begin{proposition} 
Assume that $n_1 \equiv - 1 \pmod{4}$ and $n_2 \equiv - 1 \pmod{4}$. Then 
$-1 \in D_1$ and $-1 \in D_{1}^{(n_j)}$ for all $j$. 
\end{proposition} 

We will in the sequel need the following proposition whose proof is omitted. 

\begin{proposition}\label{prop-fourd}
The integer $4 \in D_0$ and $4 \bmod{n_j} \in D_0^{(n_j)}$ for all $j$ if and only if 
$
n_2 \equiv \pm 1 \pmod{4}  
$
for all $j$. 
\end{proposition}

If $n_1 \equiv \pm 1 \pmod{4}$ and $n_2 \equiv \pm 1 \pmod{4}$, then we have the 
following factorization of $x^n-1$ over $\gf(4)$: 
$$
x^n-1=(x-1) d_0^{(n_1)}(x)  d_0^{(n_2)}(x) d_0(x)  d_1^{(n_1)}(x)  d_1^{(n_2)}(x) d_1(x),    
$$
and thus the eight quaternary cyclic codes $\D_{(i,j,h)}^{(n_1, n_2, 4)}$.  

\begin{example}\label{exam-42} 
Let $(n_1, n_2, q)=(5, 7, 4)$. In this case 
the eight quaternary cyclic codes and their minimum weights are depicted in Table \ref{tab-4D1}. 
In this example, half of the codes have minimum weight 7 and are almost the best quaternary cyclic codes of length 35 and 
dimension 18 according to Table \ref{tab-mini4}. 
\end{example}

\vspace{.25cm}
\begin{table}[ht]
\caption{The quaternary cyclic codes of length 35 and dimension 18 from the cyclotomy of Section \ref{sec-cyclDH}}\label{tab-4D1}
\begin{center}
{\begin{tabular}{|l|r|} \hline
Generator Polynomial & Minimum Weight \\\hline \hline
$\D_{(0,0,0)}^{(5, 7, 4)}$     &             7    \\ \hline 
$\D_{(1,0,0)}^{(5, 7, 4)}$   &               4 \\ \hline 
$\D_{(0,1,0)}^{(5, 7, 4)}$  &              7   \\ \hline 
$\D_{(0,0,1)}^{(5, 7, 4)}$   &               4  \\ \hline 
$\D_{(1,1,0)}^{(5, 7, 4)}$   &               4  \\ \hline
$\D_{(1,0,1)}^{(5, 7, 4)}$   &             7    \\ \hline 
$\D_{(0,1,1)}^{(5, 7, 4)}$   &               4 \\ \hline
$\D_{(1,1,1)}^{(5, 7, 4)}$  &             7    \\ \hline
\end{tabular}
}
\end{center}
\end{table}

\section{A new cyclotomy of order two and its codes}\label{sec-lastnew}

\subsection{The new cyclotomy of order two}\label{sec-cyclDD} 

Let symbols be the same as before. Since $n_1$ and $n_2$ are odd primes, $d$ must be 
even. It is easily seen that $e=(n_1-1)(n_2-1)/d$ is also even. 
Define 
\begin{eqnarray*}\label{eqn-2rdcyc}
V_0 &=& \{g^s\nu^i: 0\le s \le e-1, \ 0\le i \le d-1 \mbox{ and } s+i \mbox{ is even} \}, \\
V_1 &=& \{g^s\nu^i: 0\le s \le e-1, \ 0\le i \le d-1 \mbox{ and } s+i \mbox{ is odd} \}. 
\end{eqnarray*} 
Clearly, $V_1=gV_0$, and $V_0$ and $V_1$ form a partition of $\bZ_n^*$.

A proof of the following proposition is straightforward and is omitted. 

\begin{proposition} \label{prop-basisgroup22} 
Let symbols be the same as before. 
\begin{enumerate} 
\item $V_0$ is a subgroup of $\bZ_n^*$ and has order $(n_1-1)(n_2-1)/2$. 
\item If $a \in V_0$, we have $aV_i=V_i$. If $a \in V_1$, we have $aV_i=V_{(i+1) \bmod 2}$. 
\end{enumerate} 
\end{proposition} 

The sets $V_0$ and $V_1$ are the new cyclotomic classes of order 2, and are clearly 
different from Whiteman's cyclotomic classes of order 2 described in Section \ref{sec-W} 
and the one of Section \ref{sec-cyclDH}.

\subsection{The construction of eight cyclic codes} 

Let $\theta$ and other symbols be the same as before. 
Define 
\begin{eqnarray}
v_i(x)=\prod_{i \in V_i} (x - \theta^i). 
\end{eqnarray}

\begin{proposition} 
If $q \in V_0$ and $q \bmod n_i \in D_0^{(n_i)}$ for each $i \in \{1, 2\}$, we have $v_i(x) \in \gf(q)[x]$ 
and $d_i^{(n_j)}(x) \in \gf(q)[x]$, and 
$$ 
x^n-1=(x-1) d_0^{(n_1)}(x)  d_0^{(n_2)}(x) v_0(x)  d_1^{(n_1)}(x)  d_1^{(n_2)}(x) v_1(x).  
$$ 
\end{proposition} 

\begin{proof} 
The proof of Proposition \ref{prop-july4} is easily modified into a proof for this proposition. 
\end{proof} 

Under the conditions that $q \in V_0$ and $q \bmod n_i \in D_0^{(n_i)}$ for each $i \in \{1, 2\}$, 
let $\V_{(i,j,h)}^{(n_1, n_2, q)}$ denote the cyclic code over $\gf(q)$ with generator polynomial 
$v_i(x)d_j^{(n_1)}(x)d_h^{(n_2)}(x)$, where $(i, j, h) \in \{0,1\}^3$. The eight codes clearly have 
length $n$ and dimension $(n+1)/2$.  

\begin{proposition} \label{prop-consistent2}
There is an $\ell \in V_1$ such that $\ell \bmod{n_i} \in D_1^{(n_i)}$ for all $i \in \{1,2\}$. 
\end{proposition} 

\begin{proof} 
It is easily checked that $g$ is such an element. 
\end{proof} 

Due to Proposition \ref{prop-consistent2}, the proof of Theorem \ref{thm-srbDH} works also for the 
following theorem. 

\begin{theorem} 
For each $(i,j,h) \in \{0,1\}^3$ 
the code  $\V_{(i,j,h)}^{(n_1, n_2, q)}$ has parameters $[n, (n+1)/2]$. Let $d_{(i,j,h)}$ denote the 
minimum odd-like weight 
in $\V_{(i,j,h)}^{(n_1, n_2, q)}$. Then 
\begin{itemize} 
\item $d_{(i,j,h)} \ge \sqrt{n}$, and 
\item $d_{(i,j,h)}^2 - d_{(i,j,h)} + 1 \ge n$ if $n_1 \equiv -1 \pmod{8}$ and $n_2 \equiv -1 \pmod{8}$. 
\end{itemize} 
\end{theorem}

\subsection{The binary case}

The following proposition will be useful later. 

\begin{proposition} 
Assume that $n_1 \equiv \pm 1 \pmod{8}$ and $n_2 \equiv \pm 1 \pmod{8}$. Then 
$-1 \in V_1$ and $-1 \in D_1^{(n_j)}$ for all $j$ if and only if $n_j \equiv -1 \pmod{8}$ 
for all $j$. 
\end{proposition} 

\begin{proof} 
Let $-1=g^s \nu^i$ for some fixed $0 \le s \le e-1$ and $0 \le i \le d-1$. By definition 
$-1 \in V_1$ if and only if $s+i$ is odd.  By (\ref{eqn-nu}), 
$$ 
-1 \equiv g^{s+i} \pmod{n_1}, \ \ -1 \equiv g^{s} \pmod{n_2} . 
$$ 
The desired conclusion then follows.  
\end{proof}

We will need the following proposition in the sequel. 

\begin{proposition}\label{prop-twoddd}
The integer $2 \in V_0$ and $2 \in D_0^{(n_j)}$ for all $j$ if and only if 
$
n_j \equiv \pm 1 \pmod{8}   
$
for all $j$. 
\end{proposition} 

\begin{proof} 
Recall that $\bZ_n^*=V_0 \cup V_1$. By Proposition \ref{prop-Whiteman}, there are two integers 
$0 \le s \le e-1$ and  $0 \le i \le d-1$ such that $2=g^s\nu^i$.  
It then follows from (\ref{eqn-commonprimi}) that 
$$ 
g^{s+i} \equiv 2 \pmod{n_1} \mbox{ and } g^{s} \equiv 2 \pmod{n_2}.   
$$ 
Hence, $s$ is even if and only if  $n_2 \equiv \pm 1 \pmod{8}$.  
Note that $2 \in V_0$ if and only if $s+i$ is even. The desired conclusion then follows.  
\end{proof} 

If $n_1 \equiv \pm 1 \pmod{8}$ and $n_2 \equiv \pm 1 \pmod{8}$, then we have the 
following factorization of $x^n-1$ over $\gf(2)$: 
$$
x^n-1=(x-1) d_0^{(n_1)}(x)  d_0^{(n_2)}(x) v_0(x)  d_1^{(n_1)}(x)  d_1^{(n_2)}(x) v_1(x),    
$$
and thus the eight binary cyclic codes $\V_{(i,j,h)}^{(n_1, n_2, 2)}$.  

\begin{example}\label{exam-2d} 
Let $(n_1, n_2, q)=(7, 17, 2)$. In this case we have 
\begin{eqnarray*}
V_0 &=& \left\{ \begin{array}{l} 
1, 2, 4, 8, 9, 11, 15, 16, 18, 22, 23, 25, 29, 30, 32, 36, 37, 39, 43, 44,
46, 50, 53, 57, 58, 60, \\ 64, 65, 67, 71, 72, 74, 78, 79, 81, 86, 88, 92, 93, 95,
99, 100, 106, 107, 109, 113, 114, 116 
                 \end{array}                  
                 \right\}, \\
V_1 &=& \left\{ \begin{array}{l} 
3, 5, 6, 10, 12, 13, 19, 20, 24, 26, 27, 31, 33, 38, 40, 41, 45, 47, 48,
52, 54, 55, 59, 61, 62, 66, \\ 69, 73, 75, 76, 80, 82, 83, 87, 89, 90, 94, 96, 97,
101, 103, 104, 108, 110, 111, 115, 117, 118 
                 \end{array}                  
                 \right\}                   
\end{eqnarray*} 
and 
\begin{eqnarray*}
d_0^{(n_1)}(x) &=& x^3+x+1, \\
d_1^{(n_1)}(x) &=& x^3+x^2+1, \\ 
d_0^{(n_2)}(x) &=& x^8 + x^7 + x^6 + x^4 + x^2 + x + 1, \\
d_1^{(n_2)}(x) &=& x^8 + x^5 + x^4 + x^3 + 1, \\ 
v_0(x) &=& x^{48} + x^{47} + x^{46} + x^{44} + x^{41} + x^{40} + x^{39} + x^{37} + x^{34} + x^{33} + x^{32} +
    x^{30} + x^{27} + x^{26} + \\ 
    & & x^{25} + x^{23} + x^{20} + x^{19} + x^{18} + x^{16} + x^{14} + x^{11} +
    x^{10} + x^9 + x^7 + x^4 + x^3 + x^2 + 1 \\
v_1(x) &=& x^{48} + x^{46} + x^{45} + x^{44} + x^{41} + x^{39} + x^{38} + x^{37} + x^{34} + x^{32} + x^{30} +
    x^{29} + x^{28} + x^{25} + \\ 
    & & x^{23} + x^{22} + x^{21} + x^{18} + x^{16} + x^{15} + x^{14} + x^{11} +
    x^9 + x^8 + x^7 + x^4 + x^2 + x + 1. 
\end{eqnarray*} 
The eight binary cyclic codes and their minimum weights are depicted in Table \ref{tab-D2}. 
All the eight binary cyclic codes defined by the generalized cyclotomy of order two in this 
section have poor minimum weights. However, this does not mean that the new generalized 
cyclotomy of order two is not interesting in coding theory. In the sequel, we will see 
that some of the quaternary cyclic codes based on this new cyclotomy are the best cyclic codes. 
\end{example}

\vspace{.25cm}
\begin{table}[ht]
\caption{The binary cyclic codes of length 119 and dimension 60 from the new cyclotomy}\label{tab-D2}
\begin{center}
{\begin{tabular}{|l|r|} \hline
Generator Polynomial & Minimum Weight \\\hline \hline
$\V_{(0,0,0)}^{(7, 17, 2)}$     &             8    \\ \hline 
$\V_{(1,0,0)}^{(7, 17, 2)}$   &               4 \\ \hline 
$\V_{(0,1,0)}^{(7, 17, 2)}$  &                4   \\ \hline 
$\V_{(0,0,1)}^{(7, 17, 2)}$   &               8  \\ \hline 
$\V_{(1,1,0)}^{(7, 17, 2)}$   &               8  \\ \hline
$\V_{(1,0,1)}^{(7, 17, 2)}$   &               4    \\ \hline 
$\V_{(0,1,1)}^{(7, 17, 2)}$   &               4 \\ \hline
$\V_{(1,1,1)}^{(7, 17, 2)}$  &                8    \\ \hline
\end{tabular}
}
\end{center}
\end{table}

\subsection{The ternary case} 

The proof of Proposition \ref{prop-threed} can be slightly modified into a proof of the 
following.  

\begin{proposition}\label{prop-threed2d}
The integer $3 \in D_0$ and $3 \bmod{n_i} \in D_0^{(n_i)}$ for all $i$ if and only if 
$$ 
n_1 \equiv \pm 1 \pmod{12} \mbox{ and } n_2 \equiv \pm 1 \pmod{12} . 
$$
\end{proposition}

Hence, in the case that $n_1 \equiv \pm 1 \pmod{12}$ and $n_2 \equiv \pm 1 \pmod{12}$, 
we have indeed the eight ternary cyclic codes $\V_{(i,j,h)}^{(n_1, n_2, 3)}$. 

\begin{example} 
When $(n_1, n_2, q)=(11, 13, 3)$, the minimum nonzero weights of the eight ternary codes 
are given in Table \ref{tab-W23D}. 
Four of the eight codes have 
minimum weight 12, and the remaining four have minimum 
weight $6$.  
\end{example} 

\vspace{.25cm}
\begin{table}[ht]
\caption{The ternary cyclic codes of length 143 and dimension 72 from the cyclotomy of Section \ref{sec-cyclDD}}\label{tab-W23D}
\begin{center}
{\begin{tabular}{|l|r|} \hline
The code & Minimum Weight \\\hline \hline
$\V_{(0,0,0)}^{(11, 13, 3)}$   &          12       \\ \hline 
$\V_{(1,0,0)}^{(11, 13, 3)}$  &           6      \\ \hline 
$\V_{(0,1,0)}^{(11, 13, 3)}$  &           6      \\ \hline 
$\V_{(0,0,1)}^{(11, 13, 3)}$  &           12      \\ \hline 
$\V_{(1,1,0)}^{(11, 13, 3)}$  &          12       \\ \hline
$\V_{(1,0,1)}^{(11, 13, 3)}$  &           6      \\ \hline 
$\V_{(0,1,1)}^{(11, 13, 3)}$  &          6       \\ \hline
$\V_{(1,1,1)}^{(11, 13, 3)}$  &          12       \\ \hline
\end{tabular}
}
\end{center}
\end{table}

\subsection{The quaternary case}

The following proposition can be similarly proved. 

\begin{proposition} 
Assume that $n_1 \equiv \pm 1 \pmod{4}$ and $n_2 \equiv \pm 1 \pmod{4}$. Then 
$-1 \in V_1$ and $-1 \in D_1^{(n_j)}$ for all $j$ if and only if $n_j \equiv -1 \pmod{4}$ 
for all $j$. 
\end{proposition}

We will need the following proposition in the sequel whose proof is omitted here. 

\begin{proposition}\label{prop-fourddd}
The integer $4 \in V_0$ and $4 \bmod{n_j} \in D_0^{(n_j)}$ for all $j$  if 
$
n_j \equiv \pm 1 \pmod{4}   
$
for all $j$. 
\end{proposition}

If $n_1 \equiv \pm 1 \pmod{4}$ and $n_2 \equiv \pm 1 \pmod{4}$, then we have the 
following factorization of $x^n-1$ over $\gf(4)$: 
$$
x^n-1=(x-1) d_0^{(n_1)}(x)  d_0^{(n_2)}(x) v_0(x)  d_1^{(n_1)}(x)  d_1^{(n_2)}(x) v_1(x),    
$$
and thus the eight binary cyclic codes $\V_{(i,j,h)}^{(n_1, n_2, 4)}$.  

\begin{example}\label{exam-4d} 
Let $(n_1, n_2, q)=(5, 7, 4)$. In this case 
the eight quaternary cyclic codes and their minimum weights are depicted in Table \ref{tab-4D2}. 
Four of them are the best quaternary cyclic codes of length $35$ and dimension $18$ according 
to Table \ref{tab-mini4}. So this example shows that the new cyclotomy of order two is 
interesting in coding theory.   
\end{example}

\vspace{.25cm}
\begin{table}[ht]
\caption{The quaternary cyclic codes of length 35 and dimension 18 from the new cyclotomy}\label{tab-4D2}
\begin{center}
{\begin{tabular}{|l|r|} \hline
Generator Polynomial & Minimum Weight \\\hline \hline
$\V_{(0,0,0)}^{(5, 7, 4)}$     &             8    \\ \hline 
$\V_{(1,0,0)}^{(5, 7, 4)}$   &               4 \\ \hline 
$\V_{(0,1,0)}^{(5, 7, 4)}$  &                4   \\ \hline 
$\V_{(0,0,1)}^{(5, 7, 4)}$   &               8  \\ \hline 
$\V_{(1,1,0)}^{(5, 7, 4)}$   &               8  \\ \hline
$\V_{(1,0,1)}^{(5, 7, 4)}$   &               4    \\ \hline 
$\V_{(0,1,1)}^{(5, 7, 4)}$   &               4 \\ \hline
$\V_{(1,1,1)}^{(5, 7, 4)}$  &                8    \\ \hline
\end{tabular}
}
\end{center}
\end{table}

\section{Concluding remarks and open problems} 

Let $n$ be odd, and let $E_0$ and $E_1$ be two subsets of $\bZ_n\setminus \{0\}=\{1,2, 
\cdots, n-1\}$. Let $\mu$ be an invertible element of $\bZ_n$. 
A pair of sets $E_0$ and $E_1$, each of which is a union of nonzero 
$q$-cyclotomic cosets, forms a {\em splitting of $n$ given by $\mu$} if 
\begin{eqnarray*} 
\mu E_0=E_1, \ \mu E_1=E_0, \ E_0 \cup E_1 =\{1,2, \cdots, n-1\}. 
\end{eqnarray*}  
Clearly, if $E_0$ and $E_1$ is a splitting of $n$, we have 
\begin{eqnarray*} 
|E_0|=|E_1|=(n-1)/2. 
\end{eqnarray*} 
As before, let $\theta$ be a $n$th primitive root of unity over an extension 
field of $\gf(q)$. Define a pair of polynomials 
$$ 
g_i(x)=\prod_{j \in E_i} (x-\theta^j), \ j = 0, 1. 
$$
Since $E_i$ is the union of a number of $q$-cyclotomic cosets, each 
$g_i(x)$ must be over $\gf(q)$. The two codes of length $n$ over $\gf(q)$ 
with generator polynomials $g_0$ and $g_1$ are called a pair of duadic 
codes. The two codes of length $n$ over $\gf(q)$ 
with generator polynomials $(x-1)g_0(x)$ and $(x-1)g_1(x)$ are also   
called a pair of duadic codes. Duadic codes can also be defined in 
terms of idempotents, and include the quadratic residue codes. They were introduced 
and studied by Leon, Masley and Pless \cite{Leon84}, Leon \cite{Leon88}, 
Pless \cite{Ples86}, and Pless, Masley and Leon \cite{Ples87}, where a 
number of properties are described. Also all binary duadic codes of length 
until 241 are described in \cite{Ples87}. 

The cyclotomic cyclic codes presented in Sections \ref{sec-D} and 
\ref{sec-lastnew} should be duadic codes 
(see \cite{DLX,DP,Leon84,Leon88,Ples86,Ples87}). The contributions 
of Sections \ref{sec-D} and \ref{sec-lastnew} are the extension of an earlier 
cyclotomy of order two, the new cyclotomy of order two, and the 
cyclotomic constructions of the cyclic codes over $\gf(q)$. 
According to \cite[p. 233]{HPbook}, there are four binary duadic 
codes of length 119 and dimension $60$ and minimum weight 12.  
It is interesting that all the best binary duadic codes of length 119 
and dimension 60 are covered by the construction of Section \ref{sec-Ding}.

The contribution of Section \ref{sec-W} is the cyclotomic construction 
of the cyclic codes that contains some of the codes in \cite{DH99} as 
special cases. Experimental data shows that all the cyclic codes over 
$\gf(q)$ obtained from the cyclotomy $(U_0, U_1)$ are the best or 
almost the best cyclic codes with the same length and dimension. So it 
would be interesting to further investigate this cyclotomic construction 
of cyclic codes. Note that 
the codes of Section \ref{sec-W} are not duadic, though they are cyclotomic.    

Except for quadratic residue codes, it looks hard to develop general 
and tight lower bounds for duadic codes. The same looks true for the 
cyclotomic codes described in this paper. For specific cyclic codes 
obtained from the three constructions of this paper the BCH and other 
bounds described in \cite{LintW,EL} may be employed.  

In summary, all the three cyclotomic constructions described in this 
paper produce the best cyclic codes over certain fields $\gf(q)$, and  
are simple in structure. In addition, the extended cyclotomy and the new 
cyclotomy of order 2 may have applications in other areas. 

Finally, we mention that the cyclic codes described in this 
paper can be employed to construct secret sharing schemes \cite{CDY05}, 
authentication codes \cite{DW05} and frequency hopping sequences \cite{DFFJM}.       

\section*{Acknowledgments} 

The author is very grateful to the anonymous reviewers and Dr. Mario Blaum 
for their comments and suggestions that improved the quality of this paper.


\begin{thebibliography}{99} 

\bibitem{BS06} E. Betti and M. Sala, ``A new bound for the minimum distance 
of a cyclic code from its defining set," {\em IEEE Trans. Inform. Theory,} 
vol. 52, no. 8, pp. 3700--3706, 2006.  

\bibitem{CDY05} C. Carlet, C. Ding and J. Yuan, ``Linear codes from highly 
nonlinear functions and their secret sharing schemes,'' {\em IEEE Trans. 
Inform. Theory,} vol. 51, no. 6, pp. 2089–-2102, 2005.

\bibitem{Chie} R. T. Chien, "Cyclic decoding procedure for the 
Bose-Chaudhuri-Hocquenghem codes," {\em IEEE Trans. Inform. Theory,} 
vol. IT--10, pp. 357--363, October 1964.  

\bibitem{DFFJM} C. Ding, R. Fuji-Hara, Y. Fujiwara, M. Jimbo, and M. Mishima, ``Sets of 
frequency hopping sequences: bounds and optimal constructions,''  {\em IEEE Trans. 
Inform. Theory,}  vol. 55, no. 7, pp. 3297–-3304, July 2009.

\bibitem{DH98} C. Ding and T. Helleseth, ``New generalized cyclotomy and its applications,"  
{\em IEEE Finite Fields and Their Applications}, vol. 4, pp. 140-–166, 1998.

\bibitem{DH99} C. Ding and T. Helleseth, ``Generalized cyclotomic codes of length 
$p_1^{e_1}p_2^{e_2} \cdots p_t^{e_t}$, '' 
{\em IEEE Trans. Inform. Theory}, vol. 45, no. 2, pp. 467-–474, 1999. 

\bibitem{DP}  C. Ding and V. Pless, ``Cyclotomy and duadic codes of prime lengths,'' 
{\em IEEE Trans. Inform. Theory,} vol. 45, no. 2, 453-–466, 1999.

\bibitem{DLX} C. Ding, K. Y. Lam and C. Xing, ``Enumeration and construction of all duadic 
codes of length $p^m$,''  {\em Fundamenta Informaticae,} vol. 38, no. 1, pp. 149-–161, 1999. 

\bibitem{DW05} C. Ding and X. Wang, ``A coding theory construction of new systematic 
authentication codes,'' {\em Theoretical Computer Science,} vol. 330, no. 1, pp. 81-–99, 2005. 

\bibitem{Forn} G. D. Forney, "On decoding BCH codes," {\em IEEE Trans. 
Inform. Theory,} vol. IT-11, pp. 549--557, October 1965.

\bibitem{HPbook} W. C. Huffman and V. Pless, {\em Fundamentals of Error-Correcting Codes,} 
Cambridge University Press, Cambridge, 2003.  

\bibitem{JLX} Y. Jia, S. Ling,  C. Xing, ``On self-dual cyclic codes over finite fields," 
{\em IEEE Trans. Inform. Theory,} vol. 57, no. 4, pp. 2243--2251, 2011. 

\bibitem{Leon84} J. S. Leon, J. M. Masley, and V. Pless, ``Duadic codes,''   
         {\em IEEE Trans. Inform. Theory,} vol. IT-30, pp. 709--714, 1984. 

\bibitem{Leon88} J. S. Leon, ``A probabilistic algorithm for computing 
         minimum weight of large error-correcting codes,'' {\em IEEE Trans.  
         Inform. Theory,} vol. IT-34, pp. 1354--1359, 1998. 

\bibitem{MacW72} F. J. MacWilliams, ``Cyclotomic numbers, coding theory 
         and orthogonal polynomials,'' {\em Discrete Mathematics,} vol. 3, 
         pp. 133--151, 1972. 
         
\bibitem{MoisT} M. Moisio, ``Exponential sums, Gauss sums and cyclic codes,'' 
PhD Thesis, Acta Univ. Oul. A 306, 1998.     

\bibitem{MV99} M. Moisio and Keijo O. V{\"a}\"an\"anen,  
``Two recursive algorithms for computing the weight distribution of 
certain irreducible cyclic codes," {\em IEEE Trans. Inform. Theory,} 
vol. 45, no. 4, pp.  1244--1249, May 1999.    
         

\bibitem{Ples86} V. Pless, ``Q-codes," {\em J. Comb. Theory,}  vol. A 43, 
         pp. 258--276, 1986. 

\bibitem{Ples87} V. Pless, J. M. Masley, and J. S. Leon, ``On weights 
         in duadic codes,'' {\em J. Comb. Theory,} vol. A 44, pp. 6--21, 1987.     
         
\bibitem{Pran} E. Prange, "Some cyclic error-correcting codes with simple decoding
algorithms," Air Force Cambridge Research Center-TN-58-156, Cam-
bridge, Mass., April 1958.         

\bibitem{Rong} C. Rong, T. Helleseth, ``Use characteristic sets to decode cyclic 
codes up to actual minimum distance,'' In: {\em London Mathematical Society 
Lecture Note Series} 233, pp. 297--312, 1996.     

\bibitem{EL} M. van Eupen and J. H. van Lint, ``On the minimum 
distance of ternary cyclic codes," {\em IEEE Trans. Inform. Theory,} vol. 39, 
no. 2, pp. 409--416, March 1993. 

\bibitem{LintW} J. H. van Lint and R. M. Wilson, ``On the minimum distance 
of cyclic codes,'' {\em IEEE Trans. Inform. Theory,} vol. 32, no. 1, pp. 23--40, 
Jan. 1986.  

\bibitem{White} A. L. Whiteman, ``A family of difference sets,''   {\em Illinois J. Math.}, vol. 6, pp. 107--121, 1962. 


\end{thebibliography}
\end{document}